%% file: main.tex
\documentclass[letterpaper, 10 pt, conference]{ieeeconf}  % Comment this line out if you need a4paper

\IEEEoverridecommandlockouts                              % This command is only needed if 
\makeatletter
\long\def\@makecaption#1#2{\ifx\@captype\@IEEEtablestring%
\footnotesize\begin{center}{\normalfont\footnotesize #1}\\
{\normalfont\footnotesize\scshape #2}\end{center}%
\@IEEEtablecaptionsepspace
\else
\@IEEEfigurecaptionsepspace
\setbox\@tempboxa\hbox{\normalfont\footnotesize {#1.}~~ #2}%
\ifdim \wd\@tempboxa >\hsize%
\setbox\@tempboxa\hbox{\normalfont\footnotesize {#1.}~~ }%
\parbox[t]{\hsize}{\normalfont\footnotesize \noindent\unhbox\@tempboxa#2}%
\else
\hbox to\hsize{\normalfont\footnotesize\hfil\box\@tempboxa\hfil}\fi\fi}
\makeatother

\usepackage{dblfloatfix}
\usepackage[noadjust]{cite}
\usepackage{cite}

\usepackage{amsmath,amssymb,amsfonts,amsthm,dsfont}
\usepackage{algorithmic}
\usepackage{graphicx}
\usepackage{textcomp}

\usepackage{xcolor}
\usepackage[Symbolsmallscale]{upgreek}
\usepackage{enumerate}
\usepackage{comment}
\usepackage[caption=false]{subfig}
\usepackage{etoolbox}
\usepackage{floatrow}
\usepackage{capt-of}
\usepackage{adjustbox}
\usepackage{xr}
\newcommand{\E}{\mathbb{E}}
\newcommand{\R}{\mathbb{R}}

\newcommand{\ones}{\mathbf{1}}

\newtheorem{theorem}{Theorem}
\newtheorem{corollary}{Corollary}
\newtheorem{proposition}{Proposition}
\newtheorem{lemma}{Lemma}

\newcommand{\mM}{\mathbf{M}}

\newcommand{\mA}{\mathbf{A}}
\newcommand{\mB}{\mathbf{B}}
\newcommand{\mC}{\mathbf{C}}

\newcommand{\0}{\mathbf{0}}

\newcommand{\vx}{\mathbf{x}}
\newcommand{\vy}{\mathbf{y}}

\newcommand{\vecv}{\mathbf{v}}

\newcommand{\cG}{\mathcal{G}}
\newcommand{\cN}{\mathcal{N}}

\newcommand{\cE}{\mathcal{E}}

\newcommand{\indicator}{\mathds{1}}

\title{\LARGE \bf
Bi-SIS Epidemics on Graphs - Quantitative Analysis of Coexistence Equilibria
}

\author{Vishwaraj Doshi$^*$, Jie Hu$^*$, and Do Young Eun% <-this % stops a space
\thanks{%{\rule{8cm}{1pt}}
$^*$Equal contributors.}
\thanks{Vishwaraj Doshi is with the Data Science and Advanced Analytics team at IQVIA. The research was conducted while he was with the Operations Research Graduate Program, North Carolina State University, Raleigh, NC 27606 (NCSU). Jie Hu and Do Young Eun are with the Department of Electrical and Computer Engineering, NCSU. Email: vishwaraj.doshi@iqvia.com, \{jhu29, dyeun\}@ncsu.edu. This work was supported in part by National Science Foundation under Grant Nos. CNS-2007423, IIS-1910749, and CNS-1824518.}
}

\begin{document}

\maketitle
\thispagestyle{empty}
\pagestyle{empty}

\input{Sections/Abstract.tex}
\input{Sections/Introduction.tex}

\input{Sections/Preliminaries.tex}
\input{Sections/Single_SIS_results.tex}

\input{Sections/Bi_SIS_results.tex}

\input{Sections/Simulation.tex}
\input{Sections/Conclusion.tex}

\bibliographystyle{IEEEtran}
\bibliography{sis}

\appendices
\input{Sections/Appendix.tex}

\end{document}

%% file: Sections/Abstract.tex
\begin{abstract}
    %We study the system of two viruses of the SIS type competing over general, overlaid graphs, that is the bi-SIS model of epidemic spread. 
    We consider a system in which two viruses of the \textit{Susceptible-Infected-Susceptible} (SIS) type compete over general, overlaid graphs. While such systems have been the focus of many recent works, they have mostly been studied in the sense of convergence analysis, with no existing results quantifying the non-trivial coexistence equilibria (CE) - that is, when both competing viruses maintain long term presence over the network. In this paper, we prove monotonicity of the CE with respect to effective infection rates of the two viruses, and provide the first quantitative analysis of such equilibria in the form of upper bounds involving spectral radii of the underlying graphs, as well as positive equilibria of related single-virus systems. Our results provide deeper insight into how the long term infection probabilities are affected by system parameters, which we further highlight via numerical results.
\end{abstract}

%% file: Sections/Introduction.tex
\section{Introduction}\label{section:intro}

The study of multiple competing viruses over graph topologies has gained considerable traction in recent years \cite{Mieghem2013,R5,R6}. This is mainly because of their versatility in modeling not just infectious diseases, but also phenomena such as opposing views and opinions \cite{ruf2017dynamics} and competing products \cite{apt2011diffusion}. These phenomena, which we will commonly refer to as \textit{epidemics} or \textit{viruses}, spread over topologies such as social networks and other media platforms, word of mouth, or even human contact - often modelled as graphs with edges representing the way we connect with one another. 

Due to the relative ease of analysis, the \textit{bi-virus} model of competition between \textit{two} epidemics has seen more profound analysis \cite{prakash2012winner,liu2016analysis,R2,sahneh2014competitive,yang2017bi,Santos2015}, the underlying viruses typically being of the \textit{Susceptible-Infected-Susceptible} (SIS) type. The original (single virus) SIS model on graph was introduced to model the spread of Gonorrhea in \cite{Yorke1976}, which also provided the complete convergence characterization. Two outcomes were shown to be possible - either the virus persists over the network in the long run, when the \textit{effective infection rate} $\tau\!>\!0$ is larger than a certain threshold value $\tau^*\!>\!0$, or the virus dies out and the system converges to a healthy state when $\tau^* \!\leq\! 0$.\footnote{The effective infection rate $\tau \!\triangleq\! \beta/\delta$, where $\beta\!>\!0$ is the infection rate of the virus and $\delta\!>\!0$ stands for the recovery rate from the virus, captures the overall \textit{strength} of a virus.}

%either the virus dies out and the system converges to a healthy state when the \textit{effective infection rate} $\tau>0$ is smaller than a threshold value $\tau^*$, or the virus persists over the network in the long run when $\tau>\tau^*$ 

%or the virus persists over the network in the long run; the condition deciding the outcome being of a threshold type. The virus dies out if the \textit{effective infection rate} $\tau>0$ is smaller than a certain threshold value $\tau^*>0$, the infection going endemic otherwise. 

This threshold conditions for the single-virus SIS model were also independently rediscovered \cite{castellano2010thresholds, mieghem2013inhomogeneous}, with follow-up works \cite{van2009virus, mieghem2012comcom} being successful in establishing quantitative bounds on the long run infection probabilities/market share/influence in the case when the virus/product/opinion persists, even showing convexity of the average infection probabilities in $1/\tau$ in some cases \cite{mieghem2013inhomogeneous}. The convergence of the system itself has been proved multiple times in the literature \cite{van2011n,gray2011stochastic,li2012susceptible} utilizing techniques other than the original Lyapunov based analysis in \cite{Yorke1976}. One such convergence proof \cite{Krause_trichotomy} relies on showing that the SIS epidemic model is a \textit{monotone dynamical system} (MDS); using proof techniques that leverage the convergence properties of monotone sequences in compact sets to extract the threshold criterion.

Recently, MDS techniques were used to establish, for the first time, the complete convergence criterion for the \textit{bi-SIS} model - involving two viruses of the SIS type competing on general, overlaid graphs \cite{vdoshi2021}; providing threshold type conditions under which both viruses (which we refer to as Virus 1 and 2) die out, or one prevails over the other. More interestingly, they were used to establish necessary and sufficient conditions for the existence and global convergence of the system to the set of \textit{coexistence equilibria} (CE), where both viruses maintain presence over the network in the long run - previously an open problem \cite{yang2017bi}. Recent works \cite{R2} also improved the qualitative understanding of the CE by showing that they are always disjoint and finitely many, except for some pathological examples.\footnote{When the system parameters lie in an algebraic set of measure zero.} However, apart from a few results which are simply by-products of the techniques utilized for the convergence proofs in \cite{vdoshi2021, R2}, there is a lack of quantitative bounds on CE, and little understanding of their monotonicity properties with respect to the system parameters. 
% Quantitative bounds on CE will be of interest for strategic decisions in the medical field, for example, to control the infection probability of either of two viruses with limited medical resources.

% (In the conclusion part): The future work includes the optimization problem on $\tau_1$ and $\tau_2$ in the CE condition to minimize the upper bound in (4) while the CE condition $\tau^*_1$ and $\tau^*_2$ are coupled and dependent on both $\tau_1$ and $\tau_2$.

%However, especially when compared to what is quantitatively known for the positive equilibrium of the single virus SIS model, there is a lack of quantitative bounds on the CE fixed points, as well as a little understanding of their monotonicity properties with respect to system parameters.
 
In this paper, we provide quantitative results characterizing the behaviour of CE of the bi-SIS model on general graphs with respect to effective infection rates $\tau_1,\tau_2$ of the two competing viruses. Building upon crucial observations obtained via fixed point analysis of the bi-virus system in the MDS framework, we provide new results on the relationship between the long run probability of being infected by Virus 1 versus that of Virus 2, with regards to change in system parameters $\tau_1,\tau_2$. These results are sharper than those emerging out of mere convergence analysis, and enable us to further quantify the connection between the CE and the positive equilibrium of corresponding single-SIS models, as well as the spectral radius of the underlying graphs in the form of various upper bounds. We also briefly show via numerical results that the upper bounds are successful in capturing the trend in which the CE fixed points change with the system parameters. Our results provide a deeper understanding of how the increase (decrease) in strength of one virus affects the presence of its competitor over the network, showing that the expected decrease (increase) can be more drastic than one would expect.

%providing non-trivial results on the relationship between the long run probability of being infected by Virus 1 versus that of Virus 2. We also provide an upper bound involving infection rates $\tau_1$, $\tau_2$ and largest eigenvalues of adjacency matrices of the underlying graphs, to the average probability of being infected by either Virus 1 or 2. This upper bound is also analogous to what has previously been proven for the single virus case in \cite{mieghem2012comcom}. Our techniques build upon critical conclusions based off the analysis of the bi-virus system in the MDS framework, which has typically found use in proving mainly convergence results for dynamical systems. We thereby show the utility of this approach even in the development of quantitative results.

The rest of the papers is organized as follows. In Section \ref{section:preliminaries}, we give succinct overview on bi-SIS model with a summary of existing convergence results. Section \ref{section:single-sis results} contains our main results of the paper, with the proofs deferred to the Appendices. We then provide brief numerical results in Section \ref{section:numerical}, followed by the conclusion.

%Our main quantitative results on the CE of bi-SIS models are presented in Section \ref{section:single-sis results}. In Section \ref{section:numerical}, we present numerical results over general graphs to observe the tightness of the upper bounds derived in Section \ref{section:single-sis results}, before providing our conclusion in Section \ref{Section: Conclusion}.

%% file: Sections/Preliminaries.tex
\section{Bi-SIS Epidemic Model - A Primer}\label{section:preliminaries}

\begin{comment}
[Notation]

In this section, we briefly introduce the bi-SIS model on general graphs, which considers the single-SIS model as a special case, and summarize the main results. We specifically focus on the conditions that decide the states of two virus (survive or die out) on the entire graph in the long run.

\end{comment}

% In this section, we first provide the basic notations used throughout the paper. We then summarize the bi-SIS model of epidemic spread, listing key convergence results from the literature to set the stage for our main results.

\subsection{Basic Notations}
We use lower case, bold-faced letters to denote column vectors $\vecv\in\R^N$, and upper case, bold-faced letters to denote square matrices $\mM\in\R^{N\times N}$. We denote by $\lambda(\mM)$ the spectral radius of a non-negative matrix $\mM$. We use $\text{diag}(\vecv)$ to denote the $N\times N$ diagonal matrix with entries of vector $\vecv\in\R^N$ on the main diagonal, and $\ones$/$\0$ for all one/zero vectors with appropriate dimensions. We write $[\vx]_i$ or normal letter $x_i$ with index $i$ to represent the $i$-th entry of vector $\vx$. For vectors, $\vx\leq \vy$ means $x_i\leq y_i$ for all $i$; $\vx < \vy$ if $\vx\leq \vy$ and $\vx\neq\vy$; $\vx \ll \vy$ if $x_i < y_i$ for all $i$. Let $\cG(\cN,\cE)$ denote a general, undirected and connected graph with its adjacency matrix $\mA=[a_{ij}]$, where $a_{ij} = \indicator_{(i,j)\in\cE}$ for any $i,j\in\cN$.

\subsection{The Bi-SIS Model}\label{prelim_bi-sis}
We consider the spread of Virus 1 and 2 on overlaid graphs $\cG_1(\cN,\cE_1)$ and $\cG_2(\cN,\cE_2)$ respectively, sharing the same set of nodes $\cN$, but different edge sets $\cE_1$ and $\cE_2$ through which the respective epidemics propagate.\footnote{Using overlaid graphs with different edge sets $\cE_1$ and $\cE_2$ model the different media through which epidemics, opinions, malware and other such phenomena propagate.} At any given time, a node $i \in \cN$ is either \textit{susceptible}, or is \textit{infected by either Virus 1 or Virus 2}. If infected by Virus 1, the node infects each its susceptible neighbors with rate $\beta_1\!>\!0$, where neighbors are determined with respect to the edge set $\cE_1$ of the graph $\cG_1(\cN,\cE_1)$. Virus 2 is transmitted similarly with rate $\beta_2\!>\!0$ through the edge set $\cE_2$. Also, infected nodes recover with rates $\delta_1,\delta_2\!>\!0$ depending on whether they are infected by Virus 1 or 2 respectively. We call $\tau_1 \!\triangleq\! \beta_1/\delta_1$ and $\tau_2 \!\triangleq\! \beta_2/\delta_2$ as the effective infection rates of two corresponding viruses. The system dynamics are described by the following set of ordinary differential equations (ODEs):
\begin{equation}\label{eqn:ode_bi-sis_i}
    \begin{split}
        \dot x_i(t) &= \beta_1(1-x_i(t)-y_i(t))\sum_{j\in\cN} a_{ij}x_j(t) - \delta_1 x_i(t), \\
        \dot y_i(t) &= \beta_2(1-x_i(t)-y_i(t))\sum_{j\in\cN} b_{ij}y_j(t) - \delta_2 y_i(t)
    \end{split}
\end{equation}
for all $i \in \cN$, where $x_i(t), y_i(t)\in[0,1]$ are the probabilities that node $i\in\cN$ is infected by Virus 1 or 2 respectively at any time $t \geq 0$. Note that $x_i(t)+y_i(t) \in [0,1]$ at all time. In a matrix-vector form, \eqref{eqn:ode_bi-sis_i} can be written as
\begin{equation}\label{eqn:ode_bi-sis}
    \begin{split}
        \dot \vx &= \beta_1\text{diag}\left(\ones-\vx - \vy \right)\mA\vx - \delta_1\vx, \\
        \dot \vy &= \beta_2\text{diag}\left(\ones-\vx - \vy\right)\mB\vy - \delta_2\vy,
    \end{split}
\end{equation}
where $\mA=[a_{ij}]$ and $\mB = [b_{ij}]$ are the adjacency matrices of the overlaid graphs $\cG_1(\cN,\cE_1)$ and $\cG_2(\cN,\cE_2)$, respectively. We denote by $E\subset[0,1]^{2N}$ the set of all possible equilibria of system \eqref{eqn:ode_bi-sis}, which trivially contains $(\0,\0)$.

The \textit{single-SIS} dynamics for Virus 1 can be obtained by setting $\vy=0$ in \eqref{eqn:ode_bi-sis}, and is given by
\begin{equation}\label{eqn:ode_sis}
    \dot \vx = \beta_1\text{diag}\left(\ones-\vx \right)\mA\vx - \delta_1\vx.
\end{equation}
When $\tau_1\!>\!\tau_1^*= 1/\lambda(\mA)$, any trajectory of the system starting from $[0,1]^N \setminus \{\0\}$ converges to a positive equilibrium $\vx^*\gg\0$, otherwise they converge to $\0$ \cite{Yorke1976}. Similarly, single-SIS model for Virus 2 can be obtained by substituting $(\beta_2,\delta_2,\mB)$ for $(\beta_1,\delta_1,\mA)$ in \eqref{eqn:ode_sis}, with its positive equilibrium $\vy^*\gg\0$ when $\tau_2\!>\!\tau_2^*= 1/\lambda(\mB)$.

A preliminary result for bi-SIS epidemics \cite{yang2017bi} is that any virus which fails to satisfy its respective single-SIS survival threshold will die out in the long run; that is, Virus 1 (Virus 2) will die out irrespective of the presence of its competing virus if $\tau_1 \!\leq\!1/\lambda(\mA)$ ($\tau_2 \!\leq\!1/\lambda(\mB)$). If, at any given time $s\geq0$, a trajectory of \eqref{eqn:ode_bi-sis} enters the sets $[0,1]^N \times \{ \0 \}$ (Virus 2 dies out) or $ \{ \0 \} \times [0,1]^N$ (Virus 1 dies out), it remains in that set for all times $t>s$, and the bi-SIS model effectively reduces to a single-SIS model corresponding to the surviving virus, whose long run behaviour is governed by the single-virus convergence criterion as outlined earlier.

The non-trivial case arises when both $\tau_1\!>\!1/\lambda(\mA)$ and $\tau_2\!>\!1/\lambda(\mB)$, for which the techniques used to derive the single-SIS convergence criterion no longer apply. Specifically, both positive equilibria of the related single-virus systems $\vx^*,\vy^* \gg \0$ may exist, and it is only under this scenario when the system can possibly converge to one of (finitely) many \textit{coexistence equilibria} of the kind $(\hat \vx, \hat \vy) \gg (\0,\0)$. The complete convergence criterion derived in \cite{vdoshi2021} does include the case when $\tau_1\!>\!1/\lambda(\mA)$ and $\tau_2\!>\!1/\lambda(\mB)$, and gives the following additional conditions on $\tau_1,\tau_2$ and the respective outcomes:

\begin{itemize}
    \item[(C1)] If $\tau_1\lambda(\text{diag}(\ones\!-\!\vy^*)\mA)>1$ and $\tau_2\lambda(\text{diag}(\ones\!-\!\vx^*)\mB)\leq1$, the bi-SIS system \eqref{eqn:ode_bi-sis} converges to $(\vx^*,\0)$;
    \item[(C2)] If $\tau_1\lambda(\text{diag}(\ones\!-\!\vy^*)\mA)\leq1$ and $\tau_2\lambda(\text{diag}(\ones\!-\!\vx^*)\mB)>1$, the bi-SIS system \eqref{eqn:ode_bi-sis} converges to $(\0,\vy^*)$;
    \item[(C3)] If $\tau_1\lambda(\text{diag}(\ones\!-\!\vy^*)\mA)>1$ and $\tau_2\lambda(\text{diag}(\ones\!-\!\vx^*)\mB)>1$, the bi-SIS system \eqref{eqn:ode_bi-sis} converges to one CE fixed point $(\hat \vx, \hat \vy)\gg (\0,\0)$ in the equilibria set $E$;
\end{itemize}
Note that in (C1)--(C3), $\vx^*$ and $\vy^*$ are the single-virus fixed points as defined earlier in the subsection.

Since our focus is quantitative characterization of CE fixed points, in the rest of this paper, we will assume that $\tau_1,\tau_2$ always satisfy $\tau_1\!>\!1/\lambda(\mA)$ and $\tau_2\!>\!1/\lambda(\mB)$ and condition (C3), unless mentioned otherwise.

%% file: Sections/Single_SIS_results.tex
\section{Quantitative analysis of the Bi-SIS model}\label{section:single-sis results}

%\subsection{Bounds on Positive Equilibrium of the SIS Model}
Before presenting our results for the bi-SIS case, we give a bound on the positive equilibria for single-virus SIS models.

\begin{proposition}\label{prop:sis_upperbound}
Consider the single-virus SIS system \eqref{eqn:ode_sis}, and let $\tau_1 \!>\! 1/\lambda(\mA)$ with $\vx^* \!\gg\! \0$ being the corresponding positive, globally attractive equilibrium. Then, we have
\begin{equation}\label{eqn:fixed_point_upperbound}
    \frac{\ones^T\vx^*}{N} \leq 1-\frac{1}{\tau_1\lambda(\mA)} \leq x^*_{\max} \triangleq \max_{i \in \cN} x_i^*.
\end{equation}
\end{proposition}

This upper bound on the average infection probability $\ones^T\vx^*/N$ in \eqref{eqn:fixed_point_upperbound} has also been alluded to in \cite{mieghem2012comcom} as emerging out of the convexity of $\vx^*$ in $1/\tau_1$ \cite{mieghem2013inhomogeneous}. Here, we present a formal statement for the bound in the form of Proposition \ref{prop:sis_upperbound}, providing a more direct proof using the Fortuin–Kasteleyn–Ginibre (FKG) and Jensen's inequalities in Appendix \ref{appendix}, without the need of first showing convexity via lengthy computations. Our approach also allows us to provide the lower bound on the largest entry of $\vx^*$, the second inequality in \eqref{eqn:fixed_point_upperbound}.

For regular graphs (with degree $d$ for every node), both the inequalities in \eqref{eqn:fixed_point_upperbound} become equality since we know from Lemma 7 \cite{van2009virus} that $\vx_i^* = 1 \!-\! 1/\tau_1 d$ for each $i \in \cN$, and $\lambda(\mA) = d$. When $\tau_1$ is only slightly larger than the threshold $1/\lambda(\mA)$, intuitively speaking, the virus should not infect a large portion of the network, since it is barely strong enough to survive. The first inequality in \eqref{eqn:fixed_point_upperbound} confirms this intuition since $\ones^T\vx^*/N$ is still close to zero for such $\tau_1$, implying that the virus barely survives in the long run.
If $\tau_1$ is very large, or $1/\tau_1\lambda(\mA) \to 0$, the upper bound of $\ones^T\vx^*/N$ in \eqref{eqn:fixed_point_upperbound} gets closer to $1$ and doesn't tell much information about $\vx^*$. From \eqref{eqn:fixed_point_upperbound}, however, the node with largest infection probability has $x^*_{\max} \to 1$, showing that the virus has at least infected the `weakest' node in the network that is susceptible to infection.

%% file: Sections/Bi_SIS_results.tex
%\subsection{Bi SIS results}\label{section:bi-sis results}

\begin{comment}
We now divert our attention to analyzing the CE fixed points of the bi-SIS epidemic model on overlaid graphs. To this end, the underlying assumption for this section is that we only consider combinations of parameters $\tau_1,\tau_2$ that satisfy the CE condition, given by
\begin{align*}
    \tau_1&>1/\lambda(\text{diag}(\ones-\vy^*)\mA), ~~~\text{and} \\
    \tau_2&>1/\lambda(\text{diag}(\ones-\vx^*)\mB).
\end{align*}
Note that since $x_i^*,y_i^*\in(0,1)$  for all $i \in \cN$, these conditions imply $\tau_1>1/\lambda(\mA)$ and $\tau_2>1/\lambda(\mB)$. Therefore the single-virus survival conditions for both epidemics are always met, and the results and discussions from Section \ref{section:single-sis results} hold for the corresponding single-virus SIS models of both the viruses; with $\vx^*\gg\0$ and $\vy^*\gg\0$ being the corresponding positive equilibria to which the systems converge to globally.
\end{comment}

We now provide quantitative results for the bi-SIS CE fixed points. A CE fixed point $(\hat \vx, \hat \vy) \gg (\0,\0)$ of system \eqref{eqn:ode_bi-sis} satisfies the following equations for each $i \in \cN$:
\begin{equation}\label{eqn:bi_SIS_fixed_eq1}
    \sum_{j\in\cN}\!\! a_{ij}\hat{x}_j \!=\! \frac{\hat{x}_i}{\tau_1 (1\!-\!\hat{x}_i\!-\!\hat{y}_i)}, ~\sum_{j\in\cN}\!\! b_{ij}\hat{y}_j \!=\! \frac{\hat{y}_i}{\tau_2 (1\!-\!\hat{x}_i\!-\!\hat{y}_i)}.
\end{equation}
Analyzing these equations by first trying to show the convexity of CE in the system parameters, as done in \cite{mieghem2013inhomogeneous} for the single-virus SIS model, would be infeasible. This is because the second-order derivatives of the bi-SIS model quickly become intractable due to the highly coupled nature of the ODE system and its fixed point equations, as seen in \eqref{eqn:ode_bi-sis} and \eqref{eqn:bi_SIS_fixed_eq1} respectively. Instead, our approach is to first leverage the underlying monotonicity properties of the bi-SIS system. Apart from the bi-virus ODE system \eqref{eqn:ode_bi-sis} being MDS \cite{vdoshi2021}, i.e., the trajectories of \eqref{eqn:ode_bi-sis} preserving the ordering of the initial points, we show in the following lemma that the CE (which is the limiting state of the system) also exhibits strong monotonicity with respect to the effective infection rates $\tau_1$ and $\tau_2$.

%The following lemma gives the relationship between a CE fixed point and system parameters $\tau_1$ and $\tau_2$.

\begin{lemma}\label{lemma:x_y_tau_relationship}
Let $(\hat \vx,\hat \vy) \!\!\gg\!\! (\0,\0)$ be a CE of the bi-SIS ODE \eqref{eqn:ode_bi-sis}. For all $i \!\in\! \cN$, entries $\hat x_i$ of $\hat \vx$ increase in $\tau_1$ (decrease in $\tau_2$), while entries $\hat y_i$ of $\hat \vy$ decrease in $\tau_1$ (increase in $\tau_2$). That is,
\begin{equation*}
    \frac{\partial \hat x_i}{\partial \tau_1} \!>\! 0, ~\frac{\partial \hat y_i}{\partial \tau_1} \!<\! 0, ~~\text{and}~~~
    \frac{\partial \hat x_i}{\partial \tau_2} \!<\! 0, ~\frac{\partial \hat y_i}{\partial \tau_2} \!>\! 0.~~~~\qed
\end{equation*}
\end{lemma}
From Lemma \ref{lemma:x_y_tau_relationship}, we can see that changes in $\hat \vx$ and $\hat \vy$, caused by perturbation to any of the system parameters, are always in the opposite direction. Moreover, changes in both $\hat \vx_i$ and $\hat \vy_i$ with respect to $\tau_1$ and $\tau_2$ are strict. This form of strong monotonicity helps us establish the following result, which better captures the coupled relationship of $\hat \vx$ and $\hat \vy$ with the system parameters in \eqref{eqn:bi_SIS_fixed_eq1}.

\begin{figure}[!t]
    \centering
    \includegraphics[width = 0.75\columnwidth]{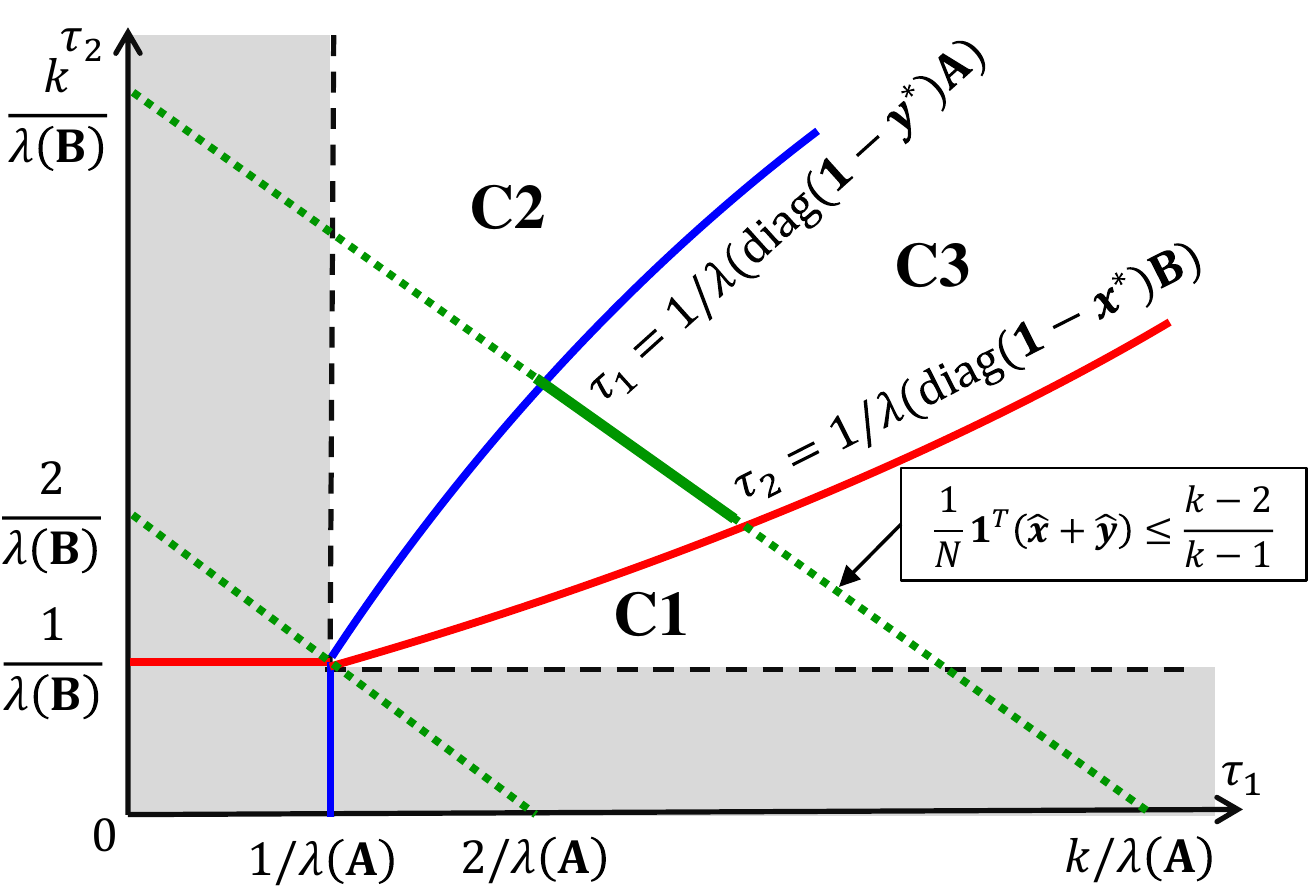}
    \caption{The unshaded region is divided by blue curve and red curve (boundary conditions of $\tau_1, \tau_2$) into three parts, corresponding to (C1)-(C3) in Section \ref{prelim_bi-sis}. The discussion on the equilibrium in the shaded region is deferred to our previous work \cite{vdoshi2021}. Green solid-line is the set of system parameters $(\tau_1,\tau_2)$ inside C3 such that they exhibit the same upper bound $(\ones^T\hat \vx + \ones^T\hat \vy)/N \leq (k-2)/(k-1)$ for $k > 2$.}
    \label{fig:my_label} 
\end{figure}

\begin{theorem}\label{thm:fi_gi_monotone}
The term $\hat y_i/(1-\hat x_i)$ strictly decreases (increases) in $\tau_1$ ($\tau_2$), $\forall i \in \cN$. Similarly, the term $\hat x_i/(1-\hat y_i)$ strictly increases (decreases) in $\tau_1$ ($\tau_2$), $\forall i \in\cN$.\qed
\end{theorem}

Lemma 1 implies that $1/(1-\hat x_i)$ increases in $\tau_1$ due to $\hat x_i$ increasing in $\tau_1$, while $\hat y_i$ decreases in $\tau_1$. However, their product $\hat y_i/(1-\hat x_i)$ may not possess any apparent monotonicity in $\tau_1$, depending on the amount of increase and decrease observed by $\hat x_i$ and $\hat y_i$. Theorem \ref{thm:fi_gi_monotone} asserts that this term indeed decreases monotonically in $\tau_1$, implying that the decrease in $1-\hat x_i$ is not large enough to offset that of $\hat y_i$ for all values of $\tau_1$ in (C3). Thus, Theorem \ref{thm:fi_gi_monotone} is much sharper in capturing the coupled change in entries of $\hat \vx$ and $\hat \vy$ as the system parameters $\tau_1$ and $\tau_2$ are varied, and we are able to do this by combining Lemma 1 with careful analysis of the first order derivatives of the CE fixed point equations \eqref{eqn:bi_SIS_fixed_eq1}. We have the following corollary as a consequence of Theorem \ref{thm:fi_gi_monotone}.

%It is not evident just from Lemma \ref{lemma:x_y_tau_relationship} that $y_i/1-x_i$, for instance, should increase of decrease with change in $\tau_1$ or $\tau_2$. 

% From Lemma (which uses MDS) we know that whatever causes x to increase will cause y to decreases and vice versa.
% However it is not obvious from this basic result itself that y/1-x should increase or decrease, and we are able to show this by combining Lemma 1 with careful analysis of the fixed point equation involving first order derivatives. Additional details are differed to the appendix.
% Loosely interpreting y_i/1-x_i as the conditional probability of node i being infected by Virus 2 given that it is not infected by virus 1, we can see that even after conditioning on the absence of Virus 1, the probability of being infected by Virus 2 is still decreases. This shows that the nature of competition in this model is stronger. Increasing tau_1 doesn't just mean that Virus 1 infects more nodes, but also that it actively reduces the infections of Virus 2, even after conditioning on its own absence at at node.

\begin{figure*}[!ht]
    \centering
    \subfloat[Using AS-733-A and AS-733-B as overlaid graphs for Corollary \ref{cor:upperbound_fi_gi}.\label{fig:my_label2} ]{\includegraphics[width=0.28\textwidth]{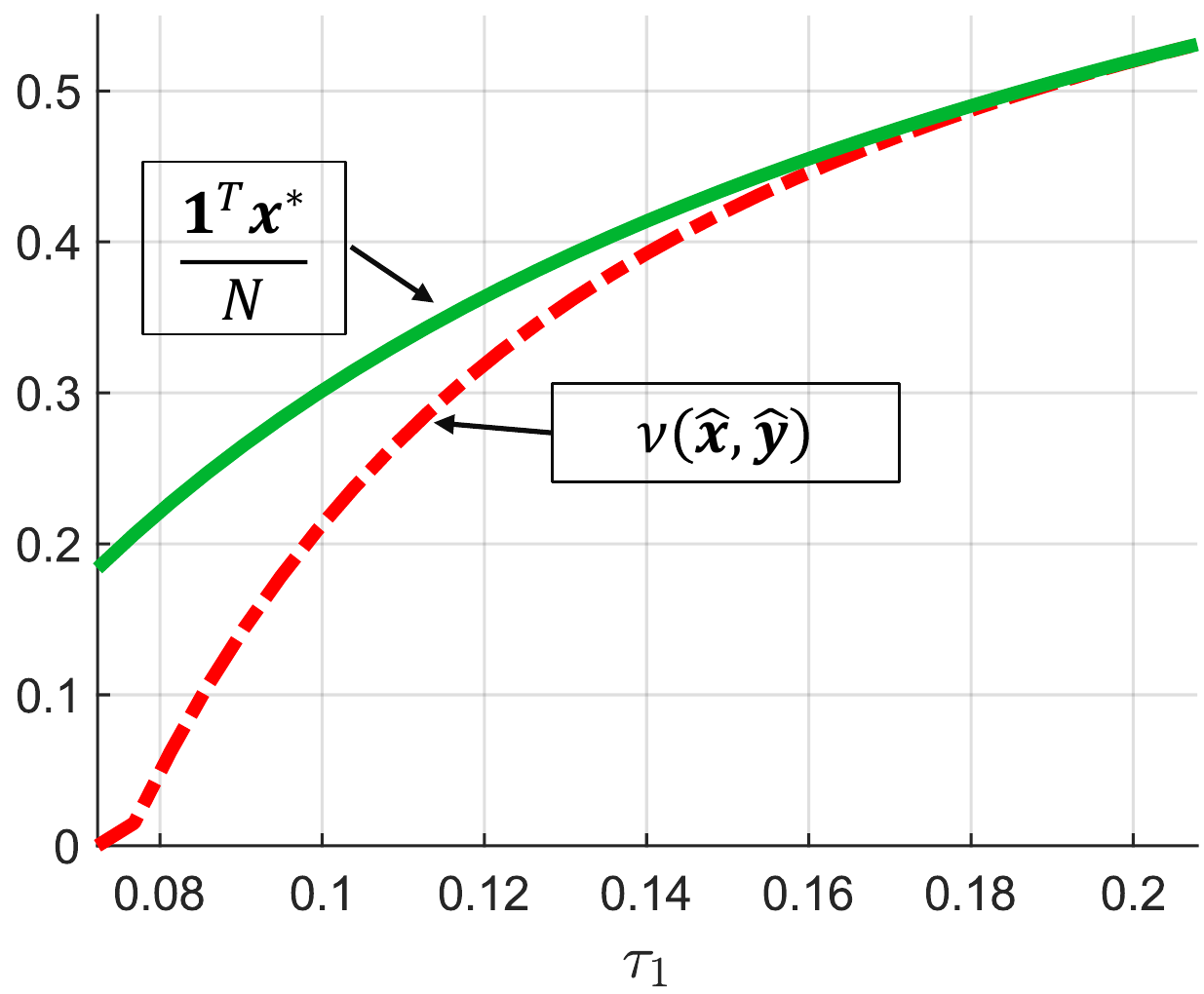}}\quad
    \subfloat[Using AS-733-A and AS-733-B as overlaid graphs for Proposition \ref{cor:upperbound_equilibrium}. \label{fig:my_label3}]{\includegraphics[width=0.28\textwidth]{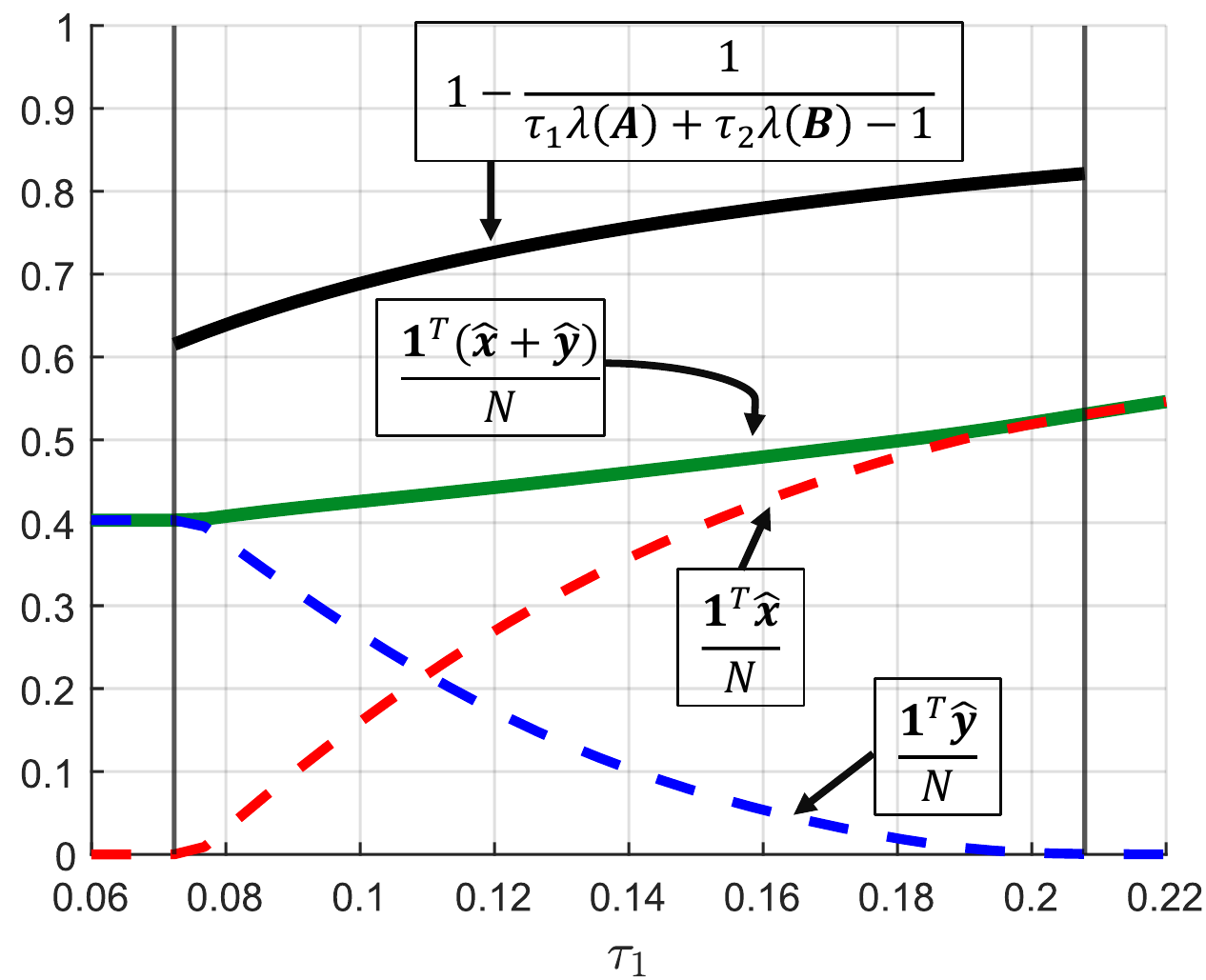}}\quad
     \subfloat[Using AS-733-B and AS-733-C as overlaid graphs for Corollary \ref{cor:upperbound_fi_gi} with three CE fixed points. \label{fig:my_label4}]{\includegraphics[width=0.32\textwidth]{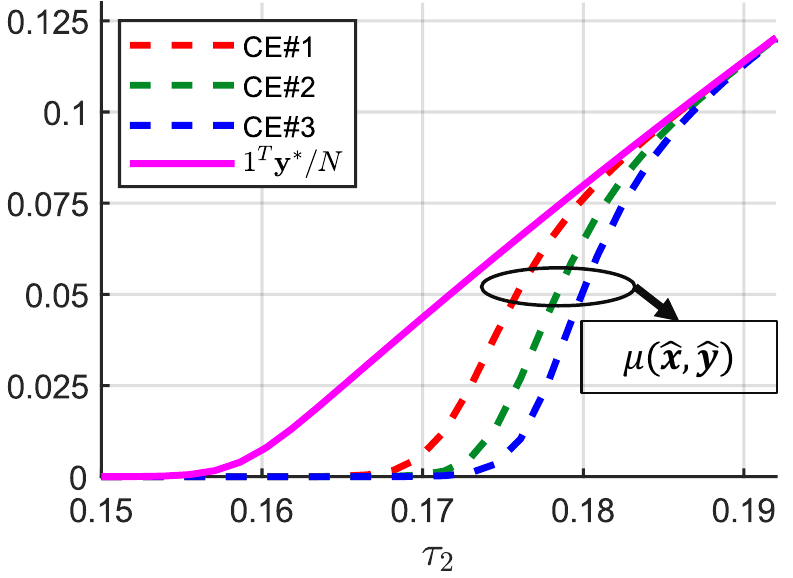}}
     \vspace{-1mm}
     \caption{Numerical results on the AS-733 graph.}
     \label{fig:myfig5}
\end{figure*}

\begin{corollary}\label{cor:upperbound_fi_gi}
For each $i \in \cN$, we have the inequalities
\begin{equation}\label{eqn:fi_gi}
    {\hat x_i} < x_i^*({1-\hat y_i}), ~~{\hat y_i} < y_i^*({1-\hat x_i}).~~~~\qed
\end{equation}
\end{corollary}

To understand the implication of Corollary \ref{cor:upperbound_fi_gi}, we briefly consider the example of competing products (modelled as viruses). Where a new product (Product 1) enters a market, more often than not, there is another existing dominant product (Product 2) enjoying its own market share $\hat \vy = \vy^*$. Through mechanisms such as marketing techniques, the Product 1 increases its own influence $\tau_1$, and eventually gains a foothold into the market $\hat \vx \gg \0$. From Lemma 1, we can only guess that the market share $\hat \vy$ would fall below its initial dominating value $\vy^*$, but there is not much one can say in terms of quantifying the reduction in $\hat\vy$. From Corollary \ref{cor:upperbound_fi_gi}, we now know that at each node $i\in\cN$, the influence of Product 1 $\hat y_i$ will fall by a factor of at least $(1-\hat x_i)$ compared to its original value of $y_i^*$. This is particularly useful when the competing viruses have access to the information about each others' local market share at each node. When this information is not available, we have the following proposition which decouples the complicated relationship between CE fixed points.

%Even though we know from Theorem 4.3 in \cite{vdoshi2021} that $\hat x_i < x_i^*$ and $\hat y_i < y_i^*$ for all $i \in \cN$, one cannot say much about how tight these bounds are. Since $\hat x_i < \hat x_i/(1-\hat y_i)$ and $\hat y_i < \hat y_i/(1-\hat x_i)$, Corollary \ref{cor:upperbound_fi_gi} improves these bounds on CE fixed points; showing that the gap between probability $\hat x_i$ ($\hat y_i$) of node $i\in\cN$ being infected by Virus 1 (Virus 2) in the presence of competition, and the infection probability $x_i^*$ ($y_i^*$) when Virus 1 (Virus 2) is the only epidemic spreading over the network, can be significant. Building upon Corollary \ref{cor:upperbound_fi_gi}, we have the following proposition.

% -> y_i < y_i*. 

% -> Theorem 1 and Corollary 1 allow us to bound the average number of nodes infected by either virus 1 or virus 2 in the long run. We present this bound in the following corollary, which can be considered as an extension of Theorem 1 to the bi-virus case.
% -> From INFOCOM, we know y_i < y_i^*. However, this result shows that y_i<y_i/1-x_i. So even after amplifying y_i, it is still less than y_i^*.

% 

\begin{proposition}\label{cor:upperbound_equilibrium}
Let $(\hat \vx, \hat \vy)$ be a CE fixed point of the bi-SIS system \eqref{eqn:ode_bi-sis}. Then, the average number of infected nodes in the network $(\ones^T\hat\vx + \ones^T\hat \vy)/N$ is upper bounded as 
\begin{equation}\label{eqn:bi-sis upperbound}
    \frac{1}{N}(\ones^T\hat \vx + \ones^T\hat \vy) < 1 - \frac{1}{\tau_1\lambda(\mA) + \tau_2\lambda(\mB) - 1}.~~~~\qed
\end{equation}
\end{proposition}

%The upper bound in Proposition \ref{cor:upperbound_equilibrium} is of a similar form to the one on $\ones^T\vx/N$ in Proposition \ref{prop:sis_upperbound}. 
Suppose that $\tau_1\lambda(\mA)$ and $\tau_2\lambda(\mB)$ are only slightly larger than $1$, implying (in light of condition (C3)) that quantities $\tau_1\lambda(\text{diag}(\ones\!-\!\vy^*)\mA)$ and $\tau_2\lambda(\text{diag}(\ones\!-\!\vx^*)\mB)$ are also only slightly larger than $1$, and just barely satisfy the coexistence condition (C3) by a small margin. In this case, the average number of infected nodes $(\ones^T\hat \vx + \ones^T \hat \vy)/N$ must also be small (albeit strictly positive), as would be expected. Note however that the upper bound in \eqref{eqn:bi-sis upperbound} holds for much larger values of $\tau_1, \tau_2$ as long as they satisfy the CE condition (C3). For instance in Figure \ref{fig:my_label}, the green dash-line represents the set of all possible $\tau_1,\tau_2$ in the CE condition satisfying $\tau_1 \lambda(\mA) + \tau_2 \lambda(\mB) = k$ such that the upper bound remains the same for all the parameters in this (level) set. We also note that the upper bound in \eqref{eqn:bi-sis upperbound} holds for all possible (finitely many) CE fixed points. In addition, our bound on the CE decouples the cross-dependency of competing viruses on overlaid graphs, into \textit{each} of single-SIS on its own graph (upper bound in \eqref{eqn:bi-sis upperbound} dependent only on the graphs, and not on $\vx^*$ or $\vy^*$). This will shed some light on how to fine-tune the system parameters $\tau_1,\tau_2$ or how tho modify the graph adjacency matrices $\mA,\mB$ for some real-world applications, e.g, strategy design in the medical area in order to control the infection probability of either of two viruses with limited medical resources.

%Tuning the system parameters $\tau_1,\tau_2$ can affect the upper bound in \eqref{eqn:bi-sis upperbound}, leading to

% information sense, corollary requires knowledge of y^*, while proposition 2 only needs the spectral radii of the underlying graphs, which can be obtained from many methods.

%% file: Sections/Simulation.tex
\section{Numerical Results}\label{section:numerical}

\begin{table}[b!] \vspace{-5mm}
\renewcommand{\arraystretch}{1.3}
\caption{System Parameters}
\label{table_example}
\centering
\begin{tabular}{|c|c|}
\hline
$\tau_1\!:\! 0.06 \sim 0.22$, $\tau_2\!=\!0.3173$ & $\lambda(\mA)\!=\!22.13$, $\lambda(\mB)\!=\!6.3$\\
\hline
$\tau_1 \!=\! 0.17$, $\tau_2\!:\! 0.15\sim 0.92$ & $\lambda(\mC) \!=\! 6.59$, $\lambda(\mB) \!=\! 6.3$\\
\hline
\end{tabular}
\end{table}

In this section, we present numerical results to assess the tightness of the upper bounds in Corollary \ref{cor:upperbound_fi_gi} and Proposition \ref{cor:upperbound_equilibrium}. To this end, we consider an undirected, connected graph AS-733 from the SNAP repository \cite{leskovec2014snap} and generate three graphs with the same $103$ nodes but with different edge sets, by modifying the edges of AS-733 while preserving the connectivity. The new graphs AS-733-A, AS-733-B and AS-733-C have $616$, $267$, and $297$ edges, respectively. Table \ref{table_example} summarizes the range of system parameters, chosen in a way to ensure that $\tau_1>1/\lambda(\mA)$ and $\tau_2>1/\lambda(\mB)$, i.e., the infection rates never lie in the gray region in Figure \ref{fig:my_label}. We numerically solve the ODE system \eqref{eqn:ode_bi-sis} for the chosen parameters until convergence is observed.

To capture how the upper bounds behave with change in system parameters, we fix $\tau_2$ and vary $\tau_1$ in Figure \ref{fig:my_label2}, \ref{fig:my_label3}. Denote by $\mu(\hat \vx, \hat \vy) \triangleq (1/N)\!\sum_{i} \hat{y}_i/(1-\hat{x}_i)$, and $\nu(\hat \vx, \hat \vy) \triangleq (1/N)\!\sum_{i} \hat{x}_i/(1-\hat{y}_i)$. In Figure \ref{fig:my_label2}, the range of $\tau_1$ and $\tau_2 $values are under the CE condition. We can see that $\nu(\hat \vx, \hat \vy)$ is increasing in $\tau_1$, as is expected from Theorem \ref{thm:fi_gi_monotone}. In addition, $\nu(\hat \vx, \hat \vy)$ gets closer to the average infection probability of Virus 1 in the single-SIS case as $\tau_1$ increases because Virus 1 becomes more dominant over Virus 2 and the bi-SIS ODE system \eqref{eqn:ode_bi-sis} behaves similar to the single-virus system. 
For Figure \ref{fig:my_label3}, as $\tau_1$ increases over the range of values from the first row in Table \ref{table_example}, the system parameters transit from regions (C1) to (C3) to (C2) of Figure \ref{fig:my_label}; Virus 1 dies out for $\tau_1 \in [0.06,0.072]$ while Virus 2 survives, both viruses survive for $\tau_1 \in (0.072, 0.208)$, and Virus 2 dies out while Virus 1 survives for $\tau_1 \in [0.208,0.22]$. The upper bound (in black line) also captures the trend for the average probability of being infected by either virus and has good estimation as $\tau_1$ increases.

Next, we fix $\tau_1$ and vary $\tau_2$, which is given in the second row in Table \ref{table_example}, in order to see how $\tau_2$, instead of $\tau_1$, can affect $\mu(\hat \vx, \hat \vy)$ in Corollary \ref{cor:upperbound_fi_gi}. Unlike to Figure \ref{fig:my_label2} and \ref{fig:my_label3} that only contain single CE, we use AS-733-B and AS-733-C as overlaid graphs to show the existence of multiple CE fixed points in Figure \ref{fig:my_label4}, which all satisfy \eqref{eqn:fi_gi} in Corollary \ref{cor:upperbound_fi_gi}. The curves in Figure \ref{fig:my_label4} also share the similar trend in Figure \ref{fig:my_label2} that the upper bound becomes tight for large $\tau_2$.

%% file: Sections/Conclusion.tex
\section{Conclusion} \label{Section: Conclusion}
In this paper we have provided, for the first time, quantitative results on the coexistence equilibria of bi-SIS epidemic models for general graphs. A future direction can include similar analysis for graphs with special topologies such as star, and line graphs, as well as cases such as ER random graphs, for which one could potentially obtain tighter results than those in Section \ref{section:single-sis results} which were presented for general graphs. \vspace{-3mm}

%% file: Sections/Appendix.tex
\section{Proofs of results}\label{appendix}

\begin{proof}[Proof of Proposition \ref{prop:sis_upperbound}]
By considering the $i$-th entry in \eqref{eqn:ode_sis} and setting $dx_i/dt = 0$, we obtain the fixed point equation\vspace{-2mm}
\begin{equation}\label{eqn:fixed_point_eq}\vspace{-2mm}
    \sum_{j\in\cN}a_{ij}x^*_j = \frac{1}{\tau}\frac{x^*_i}{1-x^*_i}.
\end{equation}
From the min-max theorem, we write $\lambda(\mA)$ in the Rayleigh quotient form such that\vspace{-2mm}
\begin{equation}\label{eqn:weightedsum_lowerbound}
\begin{split}
    \lambda(\mA) \!=\! \max_{\vy\neq\0}\!\frac{\vy^T\mA\vy}{\vy^T\vy} \!\geq\! \frac{(\vx^*)^T\mA\vx^*}{(\vx^*)^T\vx^*} \!=\! \frac{\sum\limits_{i\in\cN}x^*_i\!\!\sum\limits_{j\in\cN}a_{ij}x^*_j}{\sum\limits_{i\in\cN} (x_i^*)^2}.
\end{split}\vspace{-2mm}
\end{equation}
where the inequality comes from picking $\vy = \vx^*$ and the second equality is by rewriting the matrix multiplication in the summation notation. Define a random variable $Y$ which takes values $x^*_i$ with probability $1/N$ for all $i\in\cN$, and $\E[Y] = \ones^T\vx^*/N$. Then, replacing $\sum_{j\in\cN}a_{ij}x^*_j$ in \eqref{eqn:weightedsum_lowerbound} with \eqref{eqn:fixed_point_eq} gives\vspace{-2mm}
\begin{equation}\label{eqn:lower_bound_expect}
    \!\!\lambda(\mA) \geq \frac{\frac{1}{N}\sum_{i\in\cN} (x_i^*)^2/(1-x_i^*)}{\tau\frac{1}{N}\sum_{i\in\cN} (x_i^*)^2} = \frac{\E[Y^2/(1-Y)]}{\tau\E[Y^2]}.
\end{equation}
Since $Y^2$ and $1/(1-Y)$ are both increasing functions in $Y\in (0,1)$, FKG inequality \cite{FKG_Inequality_paper} gives $\E[Y^2/(1-Y)] \geq \E[Y^2]\E[1/(1-Y)]$.
Then, from \eqref{eqn:lower_bound_expect} we have \vspace{-1mm}
\begin{equation}\label{eqn:lower_bound_expect_form}\vspace{-1mm}
    \!\!\!\tau\lambda(\mA) \!\geq\! \frac{\E[Y^2]\E[1/(1\!-\!Y)]}{\E[Y^2]} \!=\! \E\left[\frac{1}{1\!-\!Y}\right] \!\geq\! \frac{1}{1\!-\!\E[Y]},
\end{equation}
where the second inequality comes from Jensen's inequality. Rearranging \eqref{eqn:lower_bound_expect_form} gives \eqref{eqn:fixed_point_upperbound}.

We denote $\R_+$ the set of all positive real numbers. Then, with the Collatz-Wielandt formula \cite{meyer2000matrix}, we rewrite $\lambda(\mA)$ as\vspace{-1mm}
\begin{equation}\label{eqn:cw_formulation}\vspace{-1mm}
    \!\!\lambda(\mA) \!=\! \min_{\vy \in \R_{+}^N}\!\max_{i\in\cN} \frac{[\mA\vy]_i}{y_i} \!\leq\! \max_{i\in\cN} \frac{[\mA \vx^*]_i}{x_i^*} \!=\! \frac{1}{\tau(1\!\!-\!\!x_{\max})},
\end{equation}
where the inequality is from choosing $\vy = \vx^*$. The second equality in \eqref{eqn:cw_formulation} comes from \eqref{eqn:fixed_point_eq}. Then, rearranging \eqref{eqn:cw_formulation} gives $x_{\max} \geq 1 - 1/\tau\lambda(\mA)$.
\end{proof}

\begin{proof}[Proof of Lemma \ref{lemma:x_y_tau_relationship}]
We only present the proof for the behaviour of $\hat \vx$ and $\hat \vy$ in $\tau_1$, since the case involving $\tau_2$ follows by symmetry. Consider the bi-SIS model\vspace{-1mm}
\begin{equation}\label{eqn:ode_bi-sis_epsilon}\vspace{-1mm}
    \begin{split}
        \dot \vx &= (\beta_1+\epsilon)\text{diag}\left(\ones-\vx - \vy \right)\mA\vx - \delta_1\vx \\
        \dot \vy &= \beta_2\text{diag}\left(\ones-\vx - \vy\right)\mB\vy - \delta_2\vy,
    \end{split}
\end{equation}
where we use $\epsilon$ as the parameter to vary $\tau_1$. It is enough to show that entries $\hat x_i$ of $\hat \vx$ ($\hat y_i$ of $\hat \vy$) increase (decrease) in $\epsilon$.\footnote{Instead of $\beta_1+\epsilon$, we could also vary $\tau_1$ by replacing $\delta_1$ with $\delta-\epsilon$. The steps remain similar, with both methods leading to the same conclusion.}. 

We now consider trajectories of system \eqref{eqn:ode_bi-sis_epsilon} starting from $(\hat \vx, \hat \vy)$. Note that by definition $(\hat \vx,\hat \vy) \gg (\0,\0)$ is a fixed point of the system when $\epsilon = 0$, and in this case we have $\dot \vx = 0$ and $\dot \vy = 0$. However for any $\epsilon > 0$, we have $\dot \vx \gg \0$ and $\dot \vy = \0$. Let $\phi_t(\vx,\vy)\in[0,1]^2N$ denote the flow of the system at time $t> 0$, with initial point $(\vx,\vy) \in D$, with $\phi^x_t(\vx,\vy)\in[0,1]^N$ corresponding to the the infection probabilities for Virus 1, and $\phi^y_t(\vx,\vy)\in[0,1]^N$ corresponding to those of Virus 2. Then for any $\epsilon > 0$ sufficiently small $s>0$ we have $\phi^x_s(\hat \vx,\hat \vy) > \phi^x_0(\hat \vx,\hat \vy) = \hat \vx$, and $\phi^y_s(\hat \vx,\hat \vy) = \phi^y_0(\hat \vx,\hat \vy) = \hat \vy$.
%, where the ordering $>_K$ is the south east cone-ordering introduced earlier in Section \ref{section:preliminaries}.

As a consequence of Proposition 3.1 in \cite{vdoshi2021}, the bi-SIS system is considered to be \textit{strongly monotone} in $\text{Int}(D)$.\footnote{A short educational primer on MDS is included in \cite{vdoshi2021}.} The result is that $\phi^x_{t+s}(\hat \vx,\hat \vy) \gg \vx$, and $\phi^y_{t+s}(\hat \vx,\hat \vy) \ll \vy$ for all $t > 0$. For any $\epsilon > 0$, let $(\hat \vx_\epsilon, \hat \vy_\epsilon) \triangleq \lim_{t\to\infty}\phi_t(\hat \vx,\hat \vy)$ denote the convergent point of trajectory starting from $(\hat \vx, \hat \vy)$. Since we consider only small enough $\epsilon>0$ that still satisfying the coexistence conditions (C3), the point $(\hat \vx_\epsilon, \hat \vy_\epsilon) \!\gg\! (\0,\0)$ is now the CE fixed point corresponding to the choice of $\epsilon > 0$, and satisfies $\hat \vx_\epsilon \!\gg\! \hat \vx$ and $\hat \vy_\epsilon \ll \hat \vy$. By similar arguments, the reverse holds true when $\epsilon \!<\! 0$, while still satisfying conditions (C3), that is the trajectories starting from $(\hat\vx,\hat \vy)$ converge to another CE fixed point such that $\hat \vx_\epsilon \!\ll\! \hat \vx$ and $\hat \vy_\epsilon \!\gg\! \hat \vy$.

Therefore, for all system parameters $\tau_1,\tau_2$ satisfying coexistence conditions (C3), and lying within the corresponding region in Figure \ref{fig:my_label}, the entries of $\hat \vx$ increase (entries of $\hat \vy$ decrease) in $\tau_1$ for any CE. This completes the proof.
\end{proof}

\begin{proof}[Proof of Theorem \ref{thm:fi_gi_monotone}]
Let $f_i \triangleq \hat y_i/(1-\hat x_i)$ and $g_i \triangleq \hat x_i/(1-\hat y_i)$ for notation simplicity. We only prove that $f_i$ is decreasing (increasing) in $\tau_1$ ($\tau_2$), since the similar result for $g_i$ follows by a symmetric argument. 

We use the notation by $\left[ {\partial \hat{y}_i}/{\partial \hat{x}_j} \right]_{\tau_1} \triangleq \frac{{\partial\hat{y}_i}/{\partial \tau_1}}{\partial \hat{x}_j/\partial \tau_1}$ to denote the change in $\hat y_i$ with respect to change $\hat x_i$ due to increase or decrease in $\tau_1$. Similarly, we use the notation $\left[ {\partial \hat{x}_i}/{\partial \hat{y}_j} \right]_{\tau_2}$ to denote the change in $\hat x_i$ with respect to change in $\hat y_i$ due increase or decrease in $\tau_2$. We first prove the following:\vspace{-2mm}
\begin{equation}\label{eqn:y_i_x_i_relationship}\vspace{-2mm}
    \left[\frac{\partial \hat{y}_i}{\partial \hat{x}_i}\right]_{\tau_1} \!\!\!\!< -\frac{\hat y_i}{1-\hat x_i},
    ~~\left[\frac{\partial \hat{y}_i}{\partial \hat{x}_i}\right]_{\tau_2} \!\!\!\!< -\frac{\hat y_i}{1-\hat x_i}
\end{equation}
Taking partial derivative of the logarithm of $\sum_{j} b_{ij}\hat{y}_j$ in the fixed point equation \eqref{eqn:bi_SIS_fixed_eq1} with respect to $\tau_2$, we obtain \vspace{-3mm}
\begin{equation}\label{eqn:bi_SIS_fixed_eq3}\vspace{-3mm}
        \frac{1}{\sum\limits_{j\in\cN}\!\!\! b_{ij} \hat{y}_j}\!\!\sum_{j\in\cN}\!\! b_{ij}\frac{\partial \hat{y}_j}{\partial \tau_2} \!=\! \frac{\partial \hat{y}_i}{\partial \tau_2}\!\!\left(\!\frac{1}{\hat{y}_i} \!+\! \frac{1\!\!+\!\!\left[\frac{\partial \hat{x}_i}{\partial \hat{y}_i}\right]_{\tau_2}}{1\!-\!\hat{x}_i\!-\!\hat{y}_i}\!-\!\frac{\partial \tau_2}{\tau_2\partial \hat{y}_i}\!\right)\!\!.
\end{equation}
The left-hand side of \eqref{eqn:bi_SIS_fixed_eq3} is positive since $b_{ij}\geq 0$, $\hat{x}_j \in (0,1)$, and $\partial \hat{y}_j/\partial \tau_2 > 0$ in Lemma \ref{lemma:x_y_tau_relationship} for all $i,j\in\cN$. Then, we have\vspace{-3mm}
\begin{equation}\label{eqn:bi_SIS_fixed_eq4}\vspace{-1mm}
    0 \!<\! \frac{1}{\hat{y}_i} + \frac{1+\left[\frac{\partial \hat{x}_i}{\partial \hat{y}_i}\right]_{\tau_2}}{1\!-\!\hat{y}_i\!-\!\hat{x}_i}-\frac{1}{\tau_2\partial \hat{y}_i/\partial \tau_2} \!<\! \frac{1}{\hat{y}_i} + \frac{1\!+\!\left[\frac{\partial \hat{x}_i}{\partial \hat{y}_i}\right]_{\tau_2}}{1\!-\!\hat{x}_i\!-\!\hat{y}_i},
\end{equation}
where the first inequality is from the positivity of the terms in \eqref{eqn:bi_SIS_fixed_eq3}, and the second equality comes be removing the third summand, which we can do since $\partial \hat y_i/\partial \tau_2>0$ from Lemma \ref{lemma:x_y_tau_relationship}. Then, rearranging \eqref{eqn:bi_SIS_fixed_eq4} with respect to $\left[{\partial \hat{y}_i}/{\partial \hat{x}_i}\right]_{\tau_2}$ gives us $\left[{\partial \hat{y}_i}/{\partial \hat{x}_i}\right]_{\tau_2} < -\hat y_i/(1-\hat x_i)$. Performing the same steps by differentiating by the logarithm of the fixed point equation by $\tau_1$ instead of $\tau_2$ gives us $\left[{\partial \hat{y}_i}/{\partial \hat{x}_i}\right]_{\tau_1} < -\hat y_i/(1-\hat x_i)$, proving \eqref{eqn:y_i_x_i_relationship}.
Now, in order to show $f_i$ is strictly decreasing in $\tau_1>\tau_1^*$, it is enough to show $\partial f_i/\partial \tau_1 < 0$. Taking partial derivative of $f_i$ with respect to $\tau_1$ gives\vspace{-2mm}
\begin{equation}\label{eqn:f_i_derivative}\vspace{-2mm}
    \frac{\partial f_i}{\partial \tau_1} = \frac{(1-\hat{x}_i)\frac{\partial \hat{y}_i}{\partial \tau_1} + \hat{y}_i \frac{\partial \hat{x}_i}{\partial \tau_1}}{(1-\hat{x}_i)^2} = \frac{\partial \hat{x}_i}{\partial \tau_1}\!\cdot\! \frac{(1-\hat{x}_i)\left[\frac{\partial \hat{y}_i}{\partial \hat{x}_i}\right]_{\tau_1} \!\!\!\!+ \hat{y}_i}{(1-\hat{x}_i)^2}.
\end{equation}

The first inequality of \eqref{eqn:y_i_x_i_relationship} gives $(1\!\!-\!\!\hat{x}_i)\left[{\partial \hat{y}_i}/{\partial \hat{x}_i}\right]_{\tau_1} \!\!\!+ \hat{y}_i \!\!<\!\! 0$. Together with $\partial \hat{x}_i/\partial \tau_1 > 0$ from Lemma \ref{lemma:x_y_tau_relationship}, we can see from \eqref{eqn:f_i_derivative} that $\partial f_i/\partial \tau_1 < 0$ for all $i\in\cN$. Similarly, we have\vspace{-2mm}
\begin{equation}\vspace{-2mm}
    \frac{\partial f_i}{\partial \tau_2} = \frac{\partial \hat{x}_i}{\partial \tau_2}\!\cdot\! \frac{(1-\hat{x}_i)\left[\frac{\partial \hat{y}_i}{\partial \hat{x}_i}\right]_{\tau_2} \!\!\!\!+ \hat{y}_i}{(1-\hat{x}_i)^2} > 0
\end{equation}
because $\partial \hat{x}_i/\partial \tau_2 < 0$ from Lemma \ref{lemma:x_y_tau_relationship}, showing that $f_i$ is strictly increasing in $\tau_2$ for all $i\in\cN$.
\end{proof}

%-\frac{1-\hat{y}_i}{\hat{x}_i} < 

% \begin{equation}\label{eqn:bi_SIS_fixed_eq3}
% \begin{split}
%         \frac{1}{\sum\limits_{j\in\cN}\!\! a_{ij} \hat{x}_j}\!\sum_{j\in\cN}\!\! a_{ij}\frac{\partial \hat{x}_j}{\partial \tau_1} \!&=\! \frac{1}{\hat{x}_i}\frac{\partial \hat{x}_i}{\partial \tau_1} + \frac{\frac{\partial \hat{x}_i}{\partial \tau_1} \!+\! \frac{\partial\hat{y}_i}{\partial\tau_1}}{1\!-\!\hat{x}_i\!-\!\hat{y}_i} - \frac{1}{\tau_1} \\
%         &=\! \frac{\partial \hat{x}_i}{\partial \tau_1}\!\left(\!\frac{1}{\hat{x}_i} \!+\! \frac{1+\frac{\partial \hat{y}_i}{\partial \hat{x}_i}}{1\!-\!\hat{x}_i\!-\!\hat{y}_i}\!-\!\frac{\partial \tau_1}{\tau_1\partial \hat{x}_i}\!\right)\!.
% \end{split}
% \end{equation}

\begin{proof}[Proof of Corollary \ref{cor:upperbound_fi_gi}]
We follow the notations $f_i \triangleq \hat y_i/(1-\hat x_i)$ and $g_i \triangleq \hat x_i/(1-\hat y_i)$ and assume the coexistence condition $\tau_1 > 1/\lambda(\text{diag}(\ones-\vy^*)\mA)$ and $\tau_2>1/\lambda(\text{diag}(\ones-\vx^*)\mB)$. 

If $\tau_1 \leq 1/\lambda(\text{diag}(\ones-\vy^*)\mA)$, then virus $1$ will die out, i.e., $\hat{x}_i = 0, \hat{y}_i = y^*_i$ and $f_i = y^*_i$ for all $i\in\cN$. Since $f_i$ is decreasing in $\tau_1 > 1/\lambda(\text{diag}(\ones-\vy^*)\mA)$, we have $f_i < y^*_i$. 

Similarly, if $\tau_2 \leq 1/\lambda(\text{diag}(\ones-\vx^*)\mB)$, Virus $2$ will die out, i.e., $\hat{y}_i = 0, \hat{x}_i = x^*_i$ and $g_i = x^*_i$ for all $i\in\cN$. Since $g_i$ is decreasing in $\tau_2 > 1/\lambda(\text{diag}(\ones-\vx^*)\mB)$, we have $g_i < x^*_i$.
\end{proof}

\begin{proof}[Proof of Proposition \ref{cor:upperbound_equilibrium}]
We first quantify the upper bound of $\hat{y}_i \!+\! \hat{x}_i$, $\forall i\in\cN$. Rearranging \eqref{eqn:fi_gi} gives
\begin{equation}\label{eqn:upperbound_y_hat}
    \hat{y}_i < \min\left\{y^*_i\left(1-\hat{x}_i\right), 1-\hat{x}_i/x^*_i\right\}.
\end{equation}
Then, adding $\hat{x}_i$ on both sides in \eqref{eqn:upperbound_y_hat} gives\vspace{-1mm}
\begin{equation}\label{eqn:upperbound_y_hat_x_hat}\vspace{-1mm}
    \hat{y}_i + \hat{x}_i < \min\left\{y^*_i\left(1-\hat{x}_i\right)+\hat{x}_i, 1-\hat{x}_i/x^*_i+\hat{x}_i\right\}.
\end{equation}
Note that $\hat{y}_i\in[0,y^*_i], \hat{x}_i\in[0,x^*_i]$ and $y^*_i, x^*_i\leq 1$, ensuring that the right-hand side term in \eqref{eqn:upperbound_y_hat_x_hat} is concave in $\hat{x}_i\in[0,x^*_i]$. Then, the maximum of the upper bound in \eqref{eqn:upperbound_y_hat_x_hat} is obtained by solving $y^*_i\left(1\!-\!\hat{x}_i\right)+\hat{x}_i = 1\!-\!\hat{x}_i/x^*_i+\hat{x}_i$ in terms of $\hat{x}_i$, which gives us $\hat{x}_i = (x^*_i-y^*_ix^*_i)/(1-y^*_ix^*_i)$. Putting this expression back to \eqref{eqn:upperbound_y_hat_x_hat} leads to $\hat{y}_i + \hat{x}_i < (x^*_i + y^*_i \!-\! 2x^*_iy^*_i)/(1\!-\!x^*_iy^*_i)$,
or equivalently,\vspace{-1mm}
\begin{equation}\label{eqn:1-x-y_lowerbound}\vspace{-1mm}
    1-\hat{y}_i - \hat{x}_i > \frac{(1-x^*_i)(1-y^*_i)}{1-x^*_iy^*_i} = \frac{1}{\frac{1}{1-x^*_i}+\frac{1}{1-y^*_i}-1}.
\end{equation}

Note that $\vx^*,\vy^*$ are unrelated to each other because they are the fixed points in the single-virus SIS case where the other virus dies out. We define two independent random variables $X,Y$ that take values $x^*_i,y^*_i$ with probability $1/N$ for all $i\in\cN$. From \eqref{eqn:lower_bound_expect_form}, we have\vspace{-2mm}
\begin{equation}\label{eqn:condition_rho}\vspace{-2mm}
    \tau_1\lambda(\mA) \geq \E[1/(1-X)], ~~\tau_2\lambda(\mB) \geq \E[1/(1-Y)].
\end{equation}
We also define a random variable $Z$ that takes values $1/(1-x^*_i) + 1/(1-y^*_i)$ with probability $1/N$ for all $i\in\cN$, and $\E[Z] = \E[1/(1-X)] + \E[1/(1-Y)]$. From \eqref{eqn:condition_rho}, we have\vspace{-1mm}
\begin{equation}\label{eqn:exp_Z}\vspace{-1mm}
    \E[Z] \geq \tau_1\lambda(\mA) + \tau_2\lambda(\mB).
\end{equation}
Then, summing \eqref{eqn:1-x-y_lowerbound} over all $i\!\in\!\cN$ and dividing by $N$ gives\vspace{-1mm}
\begin{equation*}\vspace{-1mm}
    \frac{1}{N}\!\sum_{i\in\cN}(1\!-\hat{y}_i - \hat{x}_i) \!>\! \frac{1}{N}\!\sum_{i\in\cN}\!\frac{1}{\frac{1}{1\!-\!x^*_i}+\frac{1}{1\!-\!y^*_i}\!-\!1} \!=\! \E\!\left[\!\frac{1}{Z\!-\!1}\!\right]\!.
\end{equation*}
Using Jensen's inequality in the above leads to \vspace{-1mm}
\begin{equation}\label{eqn:s_bar_lowerbound}\vspace{-1mm}
    \frac{1}{N}\!\sum_{i\in\cN}(1\!-\!\hat{y}_i \!-\! \hat{x}_i) > \frac{1}{\E[Z]\!-\!1} \!\geq \!\frac{1}{\tau_1\lambda(\mA) \!+\! \tau_2\lambda(\mB) \!-\! 1},
\end{equation}
where the last inequality comes from \eqref{eqn:exp_Z}. Rearranging \eqref{eqn:s_bar_lowerbound} completes the proof.
\end{proof}